\documentclass[11pt, english]{article}
\usepackage{charter,graphicx,fullpa	ge}

\usepackage{cite}
\usepackage{alltt}
\usepackage{amssymb}
\usepackage{amsthm}
\newtheorem{definition}{Definition}
\newtheorem{theorem}{Theorem}
\newtheorem{lemma}{Lemma}

\newtheorem{property}{Property}
\begin{document}
%
\title{On the periodic behavior of real-time schedulers \\
on identical multiprocessor platforms}

\author{Emmanuel Grolleau \and Jo\"{e}l Goossens \and 
Liliana Cucu-Grosjean}

%


\maketitle

\begin{abstract}
This paper is proposing a general periodicity result concerning any deterministic and memoryless scheduling algorithm (including non-work-conserving algorithms), for any context, on identical multiprocessor platforms. By context we mean the hardware architecture (uniprocessor, multicore), as well as task constraints like critical sections, precedence constraints, self-suspension, etc. Since the result is based only on the releases and deadlines, it is independent from any other parameter. Note that we do not claim that the given interval is minimal, but it is an upper bound for any cycle of any feasible schedule provided by any deterministic and memoryless scheduler.
\end{abstract}


%

\maketitle

\section{Introduction}
\subsection{Feasibility and simulation intervals}
Real-time systems are widely used nowadays; their correctness has to meet functional and temporal requirements. Real-time scheduling theory focuses on the temporal validation of such systems. The validation of a real-time system relies on a worst-case behavior: each task is characterized by some temporal properties and constraints that have to be met by the scheduling algorithm. Most task models are based on the model initially defined in~\cite{Liu1973}.

One of the key problems in real-time scheduling is the schedulability problem~\cite{carpenter}\cite{davis}: a task system is schedulable by a scheduling algorithm if, in the worst-case scenario(s), all the temporal constraints are met. For some kinds of task systems, the worst-case scenario is easy to build. As an example, the worst-case scenario is the initial instant, known as the critical instant, for independent synchronous task systems (i.e. tasks are released at the same time), or for non-concrete task systems (i.e. the release times can occur anytime), executed on a uniprocessor platform. In this context, the worst-case response time of the tasks is encountered in the first synchronous busy period. This allows some efficient feasibility tests to exist, exact in such context, running in pseudo-polynomial time (e.g.~\cite{JP86}\cite{Lehoczky1990} for fixed-task priority scheduling like Deadline Monotonic or~\cite{JS93,BHR90} for fixed-job priority scheduling like Earliest Deadline First). 

However, the critical instant does not correspond to the first synchronous busy period when the tasks are asynchronous. In this case, a large time interval has to be considered in order to reach a cycle in the schedule, and the feasibility problem is NP-hard in the strong sense~\cite{Leung1982}. The only known exact feasibility tests are simulation-based like in~\cite{Goossens1999} and have to consider the whole schedule, which is infinite, therefore, we need a finite simulation interval.
\begin{itemize}
	\item Simulation interval: gives an exact or upper bound of the time interval for the schedule to repeat in a cycle.
	\item Feasibility interval: a finite interval $[a,b]$ such that if all the deadlines of jobs released in the interval are met, then the system is schedulable.
\end{itemize}

In general, a simulation interval can be used as a feasibility interval, and a feasibility interval can be smaller than a simulation interval. For example, in uniprocessor systems, the most known feasibility interval is the first busy period of a synchronous releases of the tasks when a fixed-priority (FPP) or EDF is used. If the system is schedulable, then this interval is smaller than the simulation interval which is the least common multiple of the periods of the tasks, a.k.a. hyperperiod.

Feasibility intervals are widely used to validate synchronous task systems executed on a single processor, but in the case of asynchronous task systems, the simulation intervals can be used under some conditions as feasibility intervals. The conditions when a simulation, or a test, can use a periodicity result as a feasibility interval rely on the properties of sustainability. This notion is defined in the next section, but the general idea is that in some contexts (e.g. critical sections involved) considering the worst-case execution time of a task in a simulation does not give the worst-case behavior (we say that the system is not C-sustainable). Reducing the duration of a task may worsen the response-time of a task in the system (this phenomenon is called a scheduling anomaly). This means that, in this context, a simulation of the behavior of the system cannot be used as a schedulability test. Nevertheless, even if in the considered context, some anomalies can occur, some schedulability tests can be adapted by taking some factor into account. For example, in the case of critical sections, when using a priority ceiling protocol~\cite{SRL90}, we can compute the maximal blocking factor due to lower priority tasks, and add it in the schedulability test, making this test C-sustainable.

The simulation can also be used in general to exhibit an incorrect behavior of the system. In this case, which is corresponding usually to non-sustainable cases, an incorrect sequence can be used to prove that the system is not schedulable. However, when searching such a counter-example, we also need to know how long the system must be simulated, and therefore we need a simulation duration. If the simulation duration is known for structural constrained task systems, it would be possible for timing analysis tools building a simulation to represent a specific behavior of the schedule, even when the tasks share resources, like in~\cite{brandenburg}\cite{nima}.

A periodicity result is also helpful for offline methods computing a feasible schedule which is meant to be repeated online infinitely. It is also helpful when characterizing the behavior of a system: statistics concerning the number of preemptions, migrations, response time variations, etc., have to be given on the cyclic part of a schedule.

\subsection{Tasks systems and scheduling algorithms}

We can classify the task systems using some relevant properties. A task system is said concrete if its first release times (offsets) are known, and if the tasks are periodic. It is non-concrete if the offsets are unknown (e.g. first release triggered by an external event), or if it is sporadic (i.e. the period is giving a minimal inter-arrival time for the tasks). On a worst-case scenario point of view, for uniprocessor systems, the fact that a system is non-concrete or sporadic is equivalent, because in both cases, the critical instant (i.e. synchronous release of the tasks) can exist during the life of the system. In a concrete task system, tasks are said synchronous if all their first releases are simultaneous, and asynchronous otherwise.

For non-concrete or sporadic task systems, only online scheduling algorithms can be used: the scheduling algorithm, usually priority driven, chooses to execute the ready job(s) with the highest priority. Concrete synchronous and asynchronous systems can be scheduled using an online scheduling algorithm or an offline (time-driven) schedule built to meet the temporal requirements (using a branch and bound, an enumeration algorithm, linear programming, etc.). It is important to note that, when executed online, the system may differ compared to its task model: the actual execution time may be lower than the WCET (Worst-Case Execution Time), and for sporadic task systems, the inter-arrival time may be higher than the period parameter.

Since the question of the simulation duration arises only when the tasks are time-driven (strictly periodic tasks), and that usually a time-driven task is released by the internal clock, we consider that the periods and offsets are natural numbers, which are in reality multiples of the internal clock granularity. A task system $S$ is a set of tasks $\tau_{i_{,i=1..n}}$, defined by:
\begin{itemize}
	\item $O_i \in \mathbb{N}$ the task offset, is the release date of the first job of $\tau_i$,
	\item $C_i \in \mathbb{R}$ the Worst-Case Execution Time (WCET) is the maximum amount of time the CPU has to spend to execute a job of $\tau_i$,
	\item $T_i \in \mathbb{N}$ the (strict) task period, the jobs are released at the instants $O_i+kT_i, k\in\mathbb{N}$,
	\item $D_i \in \mathbb{R}$ is the relative deadline and represents the timing constraint of a task: the $k^{th}, k\in\mathbb{N}$ job $j_{i,k}$ must be executed in the window $[O_i+kT_i,O_i+kT_i+D_i)$. If $\forall i \in \{1..n\}, D_i\leq T_i$ the system has constrained deadlines, else the system has arbitrary deadlines,
	\item $a_{i,j}=O_i+jT_i$ the activation time of the job $\tau_{i,j}$,
	\item $d_{i,j}=O_i+jT_i+D_i$ the absolute deadline of $\tau_{i,j}$,
	\item $H$ is the hyperperiod of a system $S$ given by $\textrm{lcm}(T_1,...,T_n)$ where $\textrm{lcm}$ is the least common multiple,
	\item $O^{\max}$ is the largest offset, $O^{\max}=\max_{i=1..n}(O_i)$,
	\item $U=\sum_{i=1}^{n}{C_i/T_i}$ is the processor utilization factor.
\end{itemize}

\begin{definition}[Predictable algorithm]
\label{def:predictable}
A scheduling algorithm is predictable for a class of task systems if reducing the duration of a task cannot delay the ending date of any job in the system compared to the ending date of the system where the tasks duration is the WCET.
\end{definition}

Predictability has been recently extended to most task parameters, and is included in the more general property of sustainability.
\begin{definition}[Sustainability] 
\label{def:sustainable}
Sustainable scheduling algorithm and sustainable schedulability test:
\begin{itemize}
\item A scheduling algorithm is sustainable for a class of task systems if improving some task parameter (reducing a WCET C, increasing a period T, increasing a deadline D) cannot delay the ending date of any job in the system compared to the original task system. We speak of C-sustainability, T-sustainability, D-sustainability, respectively.
\item A schedulability test is sustainable if improving some task parameter cannot invalidate a positive result of the test. Note that even if the underlying algorithm is not sustainable in a context, a corresponding schedulability test can be designed to be sustainable.
\end{itemize}
\end{definition}

The most common task parameters variations concern the execution time because conditions and loops are data-dependent code, as well as the efficiency of processor optimizations, like cache hit/miss, instructions prefetch, and pipelines. In sporadic task system, the inter-arrival time, even if it is called a period, is also a variable parameter.

For uniprocessor systems, most online scheduling algorithms are C-sustainable (i.e. predictable), for example~\cite{Cucu2006} shows that work-conserving (i.e. never let the processor idle if there is at least one ready job) algorithms are predictable for independent task systems. This property is important since for predictable algorithms in a class of task systems, the simulation can always be used in order to validate the system. However, work-conserving scheduling algorithms are not predictable when some tasks are not preemptable, or there are shared resources, or there are precedence constraints using synchronization mechanisms~\cite{Forget2010}. In the case of multiprocessor systems, the scheduling algorithm global-EDF is not T-sustainable.

Since we focus on the periodic behavior, only deterministic and memoryless scheduling algorithms are considered. We call a fixed-task priority scheduling algorithm an algorithm giving a unique priority to all the jobs of the same task (what is called fixed-priority policy (FPP) in the literature, like Rate Monotonic, or Deadline Monotonic or their global variants like RM-us~\cite{andersson}), and a fixed-job priority scheduling algorithm is a policy assigning a fixed priority to the jobs (e.g. Earliest Deadline First - EDF).

\subsection{Existing periodicity results in uniprocessor context}

The notion of pseudo-work-conserving scheduling, introduced in~\cite{CG04}, is used in the uniprocessor context to characterize the periodicity of schedules for a very wide class of scheduling algorithms. Unfortunately, we show in Section~\ref{sec:2} that this property cannot be used anymore in the multiprocessor case.

\begin{definition}[Pseudo-work-conserving scheduler~\cite{CG04}]
A pseudo-work-conserving algorithm is a scheduling algorithm allowed to insert idle slots, but not more than $H(1-U)$ idle slots can be purposely introduced per window of size $H$. This class includes every work-conserving algorithms, but also non work-conserving algorithms controlling the amount of inserted idle slots, for example using a task modeling the idle slots.
\end{definition}

The class of pseudo-work-conserving algorithm has been the largest scheduling class, in our knowledge, studied for the problem of periodicity. This class is including every popular scheduling algorithm, but is also including many specific offline scheduling algorithms.

Note that the necessary condition $U\leq 1$ has to hold in the following results. There are two kinds of methods used to calculate the simulation interval:
\begin{itemize}
	\item Processor utilization (equivalently idle slots) based: 
	\begin{itemize}
		\item The seminal work of~\cite{Leung1980} shows that $[0,O^{\max}+2H)$, the $2^{nd}$ hyperperiod being the cyclic part of a schedule, is an upper bound of the simulation interval for fixed-task priority schedulers, and independent task systems with constrained deadlines (i.e. $D_i\leq T_i$).
		\item It is shown in~\cite{Goossens1999} that, with arbitrary deadlines, $[0,O^{\max}+2H)$ is still giving an upper bound of the simulation interval for Earliest Deadline First, and FPP scheduling algorithms.
		\item The most general result concerning task systems with constrained deadlines is given in~\cite{CG04}. It shows how to determine the minimal simulation interval for any deterministic memoryless pseudo-work-conserving scheduling algorithm. The cyclic behavior of a schedule starts exactly at the date $\theta_c$, date following the last acyclic idle slot, thus the simulation interval is given by $[0,\theta_c+H)$. The date of the last acyclic idle slot is $0\leq \theta_c\leq O^{\max}+H$. This result has been extended to non-preemptible parts, precedence constraints, and resource sharing (note that these three extensions can be used to validate a system only in the case of offline scheduling because most online scheduling algorithms are not predictable in these contexts). This result has been extended to multi-threaded tasks in~\cite{Bado12}.
	\end{itemize}
	\item Priority based: an upper bound to the simulation interval is $[0,S_n+H)$~\cite{Goossens1997} for fixed-task priority scheduling algorithms, for independent tasks with constrained deadlines, where $S_n$ is calculated iteratively on the system, giving the tasks ordered by priority level: 
\begin{eqnarray}\label{eq:sn}
S_1&=&O_1 \\
S_i&=&\max (O_i,O_i+\left\lceil \frac{S_{i-1}-O_i}{T_i}\right\rceil T_i) \nonumber
\end{eqnarray}
\end{itemize}

We can notice that the feasibility or simulation interval problem for arbitrary deadlines systems is still an open problem in the case of any algorithm other than EDF or FPP: this paper will fill this gap with an upper bound. 

\subsection{Existing periodicity results in a multiprocessor context}

The periodic behavior of schedulers has been studied in the context of global scheduling on multiprocessor platforms for specific scheduling algorithms. For partitioned scheduling, as long as there is no migration, the simulation duration problem consists in studying the simulation duration on each processor: this is thus related to the uniprocessor problem. In the sequel, we consider the problem of the periodic behavior of global schedulers.

Every known result concerning global scheduling is concerning independent task systems. There are several periodicity results in~\cite{Cucu06} concerning constrained deadline systems, on uniform multiprocessor systems, that can be applied to the identical multiprocessor platforms. If the tasks are synchronous, then any feasible schedule generated by a deterministic and memoryless scheduler has a periodic behavior on the interval $[0,H)$. For asynchronous tasks systems, $[0,S_n+H)$ is a simulation interval of any feasible schedule generated by a FPP scheduler for asynchronous systems, using the same $S_n$ as in~\cite{Goossens1997}. 

The case of arbitrary deadlines systems has been studied in~\cite{CucuG07} for identical multiprocessor platforms. It is shown that any feasible schedule generated by a deterministic and memoryless scheduler is finally periodic. Moreover, for a feasible schedule generated by a FPP scheduler, $[0,H)$ is a simulation interval for synchronous systems, while $[0,\hat{S}_n+H]$ is a simulation interval for asynchronous systems, with, giving the tasks ordered by priority order:
\begin{eqnarray}\label{eq:sn2}
\hat{S}_1&=&O_1 \\
\hat{S}_i&=&\max \left(O_i,O_i+\left\lceil \frac{\hat{S}_{i-1}-O_i}{T_i}\right\rceil T_i\right)+H_i \nonumber
\end{eqnarray}
with $H_i=\textrm{lcm}_{j=1..i}(T_i)$.
This result has been extended to the case of unrelated multiprocessor platforms in~\cite{Cucu-GrosjeanG11}.

\paragraph*{This research}
This paper is the first result concerning the simulation interval applicable to a very large context: on identical multiprocessor platforms, for any deterministic and memoryless scheduler, scheduling asynchronous periodic tasks with arbitrary deadlines, subject to a large class of structural constraints (including precedence constraints, mutual exclusions, self-suspensions, preemptive or non-preemptive tasks). The results concerning multiprocessor platforms currently known in the literature all consider independent and preemptive periodic tasks scheduled by specific schedulers (in our knowledge, the global versions of FPP and EDF only). An interesting intermediate result, Lemma~\ref{lem:riDi}, shows that, for this problem, the synchronous case is a worst-case scenario.

\subsection{Organization of the paper}

In Section~\ref{sec:2}, we show on a simple motivating example that for synchronous task systems on multiprocessor system, the hyperperiod cannot be considered as a simulation interval for arbitrary deadline systems. In Section~\ref{sec:result}, we show that non-work-conserving scheduler have to be taken into account because even popular schedulers (such as global-EDF or global-FPP) do not behave as work-conserving scheduling algorithms for multiprocessor systems. Then we show that the set of feasible schedules for asynchronous task systems are included in the set of feasible schedules for synchronous arbitrary deadlines systems. This is allowing us to easily prove our general result which is Theroem \ref{th:general}. Section~\ref{sec:discuss} is a discussion about the possible ways to improve our upper bound.

\section{Motivational example}
\label{sec:2}
The periodicity results, for constrained deadlines synchronous systems, use the fact that the first time window of size $H$, where exactly $H/T_i$ jobs of each task $\tau_i$ are completely executed, is giving the cycle of a schedule, and that a feasible schedule has to conform to this pattern because the amount of releases and corresponding deadlines is exactly $H/T_i$ in this time window. Nevertheless, when a deadline $D_i$ is greater than the period $T_i$, we can imagine a feasible schedule spilling on the second hyperperiod (i.e. executing less than $H/T_i$ jobs for a task in a hyperperiod), thus inserting \emph{more} than $H(1-U)$ idle slots in this time window, and catching up in a second hyperperiod by inserting less than $H(1-U)$ idle slots, while still meeting the deadlines of the tasks. One can wonder what could be the advantages of such a policy, and just ignore this phenomenon by arguing that no algorithm purposely delays a schedule by inserting an acyclic idle slot. Nevertheless, this phenomenon can happen in schedules generated by popular work-conserving scheduling algorithms (e.g. global-EDF) for multiprocessor systems. This is why we need to understand it.

\subsection{Periodicity of work-conserving algorithms in the uniprocessor case}

In uniprocessor systems, for popular work-conserving schedulers like fixed-task priority schedulers or the EDF scheduler, we cannot observe any task spilling over another hyperperiod, as it is stated in the two following theorems.

\begin{theorem}
\label{th:edfarbitrary}
\cite{Goossens1999} In the context of asynchronous arbitrary deadline systems, with $U\leq 1$, on a single processor, a feasible schedule built by an EDF scheduler has a transient phase of at most $O^{\max}+H$ and a steady phase of size $H$.
\end{theorem}

\begin{theorem}
\label{th:fpparbitrary}
\cite{Goossens99} In the context of asynchronous arbitrary deadline systems, with $U\leq 1$, on a single processor, a feasible schedule built by a fixed-task-priority algorithm has a transient phase of at most $S_n$ and a steady phase of size $H$, with $S_n$ obtained with Eq.~\ref{eq:sn}.
\end{theorem}

In their global version, however, these scheduling algorithms do not behave as work-conserving algorithms. This is resulting in cases where a feasible schedule can have a cycle greater than a hyperperiod, as illustrated in the following motivating examples.
  
\subsection{Work-conserving algorithms in the uniprocessor case are not pseudo-work-conserving in the multiprocessor case}

As an example of how popular scheduling algorithm can behave as non-pseudo-work-conserving algorithms in the multiprocessor case, we consider a very counter-intuitive example. Let $Sys_1$ be a system containing three synchronous tasks executed on two processors: $\tau_1$, characterized by $O_1=0$, $C_1=1$, $T_1=D_1=2$, $\tau_2$, with the same parameters as $\tau_1$, and $\tau_3$ with $O_3=0$, $C_3=3$, $T_3=4$ and $D_3=7$. Note that the processor utilization factor of $Sys_1$ is $U_3=7/4$, the hyperperiod is $H=4$, and the task $\tau_3$ has a deadline greater than its period. The global-EDF schedule of the system $Sys_1$ is shown in Fig.~\ref{fig:edfmulti}. We can notice that during the two first hyperperiods (i.e. in the time interval [0,8)), there are two idle slots per hyperperiod. Giving the number of processors $m=2$ and the utilization factor, we know that in a steady state of period $kH, k\in \mathbb{N}$, the amount of idle slots is exactly $kH(m-U)$. For $Sys_1$, $U=7/4$ and $m=2$, thus there must be exactly $k$ idle slots in the cyclic part of a schedule of length $kH$. We can observe that the states of the system at the date $8$ and at the date $12$ are only given by the release of every task of the system, hence, the steady state of the schedule is given by the time interval $[8,12)$, while the transient phase, despite the fact that the tasks are synchronous, lasts during $8$ time units.

\begin{figure}[!t]
\centering
\includegraphics[width=3.5in]{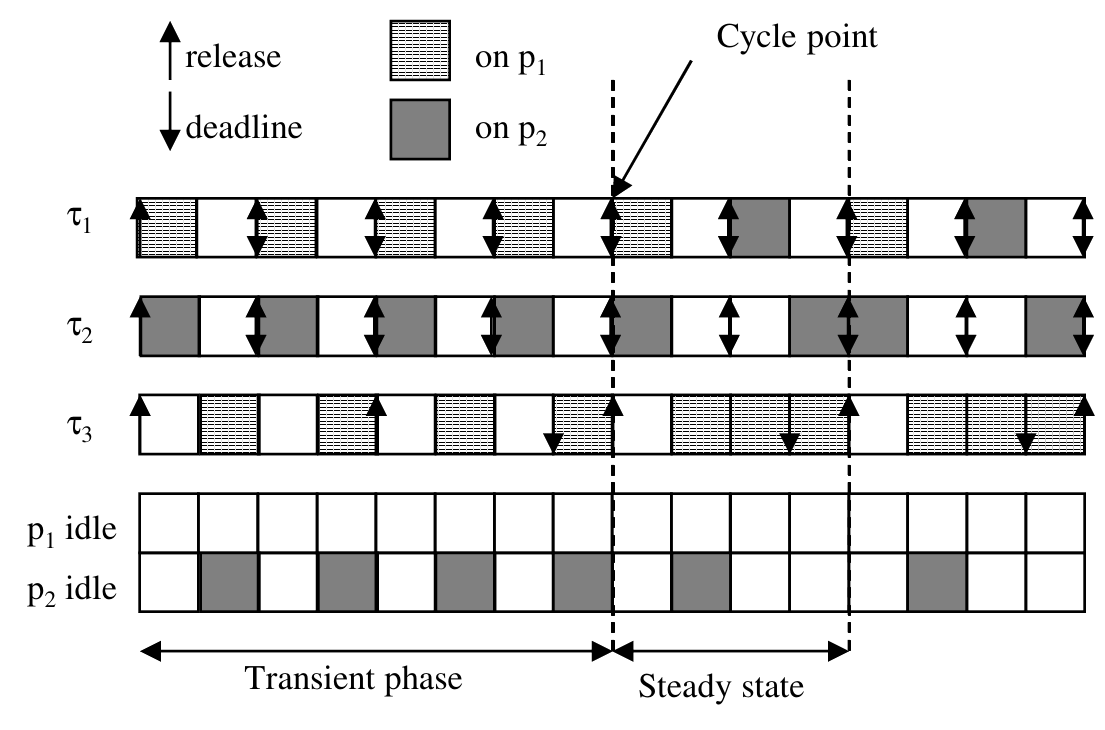}
\caption{Global-EDF schedule of $Sys_1$ on two processors}
\label{fig:edfmulti}
\end{figure}

If we consider the Longest Remaining Processing Time First (LRPTF) scheduling algorithm, the schedule of the same system has an empty transient phase (see Fig.~\ref{fig:lrptf}).

\begin{figure}[!t]
\centering
\includegraphics[width=2in]{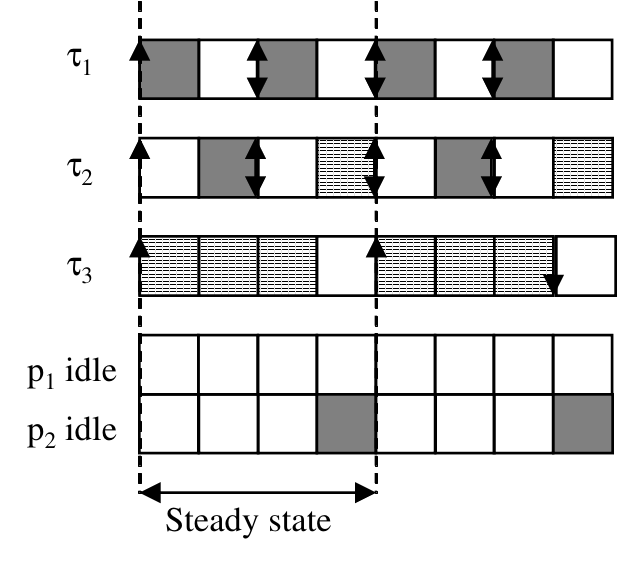}
\caption{LRPTF schedule of $Sys_1$ on two processors}
\label{fig:lrptf}
\end{figure}

We can conclude that global-EDF does not act as a pseudo-work-conserving algorithm since it is inserting more than $H(2-U)=1$ idle slots during the transient phase, while it is possible to insert only this amount of idle slots. The notion of pseudo-work-conserving algorithms has been introduced in~\cite{CG04} in the case of uniprocessor systems in order to generalize the simulation duration results to a wide class of schedulers, including every popular schedulers. We can see with this example that this property cannot be used anymore in the multiprocessor case. Therefore, a general result concerning every deterministic and memoryless scheduler, without idle times conserving constraints, is required. Such a result is the main contribution of this paper.

A Deadline Monotonic priority assignment (see Fig.~\ref{fig:dm}) of the system $Sys_1$ is not feasible, and never enters in a cycle where the good amount of idle slots is present. Since there are always two idle slots instead of one per hyperperiod, the lateness of the task $\tau_3$ is increasing with every hyperperiod. The transient phase of the schedule never ends, and in the third hyperperiod of the system, at the time instant $11$, $\tau_3$ misses its deadline.

\begin{figure}[!t]
\centering
\includegraphics[width=3.5in]{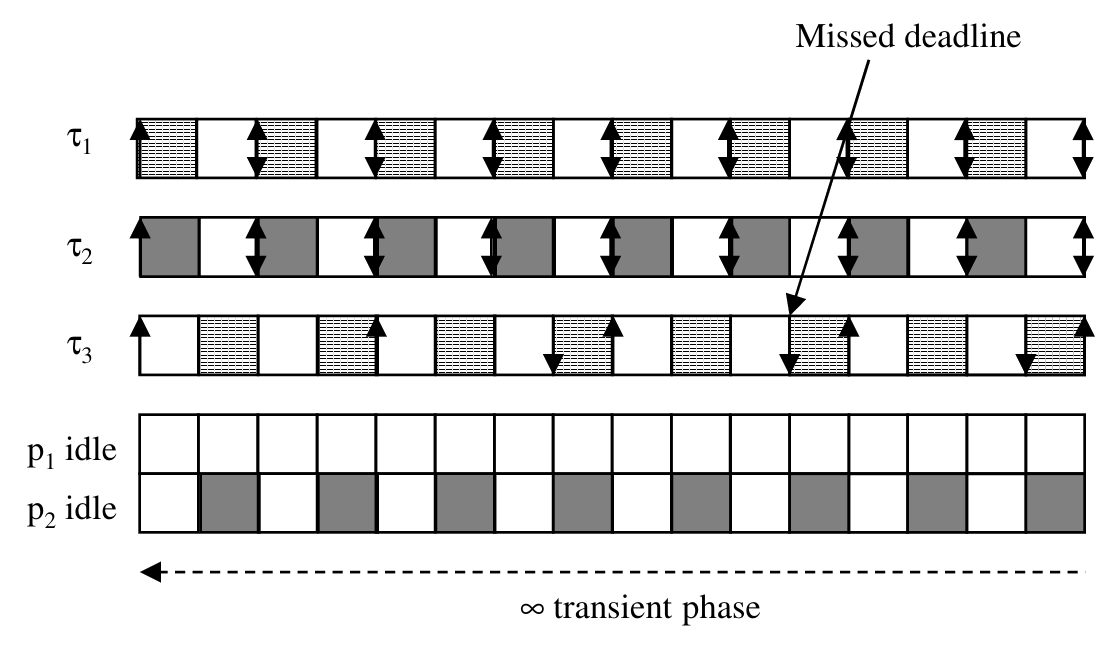}
\caption{Deadline Monotonic schedule of $Sys_1$ on two processors}
\label{fig:dm}
\end{figure}

The example of $Sys_1$ shows that, on a very simple case, a fixed-job priority scheduling algorithm (global-EDF), and a fixed-task priority scheduling algorithm (Deadline Monotonic), which are work-conserving for uniprocessor systems, act as non-pseudo-work-conserving algorithms in the multiprocessor case. It is showing also that in the case of synchronous systems with arbitrary deadlines, the time window $[0,H)$ cannot be used as a simulation interval.

\section{General periodicity result}
\label{sec:result}
This section investigates the behavior of any deterministic and memoryless scheduling algorithm. Note that we consider non-reentrant schedules, i.e. the subsequent jobs of the same task are executed in their release order. This behavior is often assured by the way the tasks are programmed on real-time operating systems. We define formally the required concepts.

\begin{definition}[State and pre-state of a system]
\label{def:state}
The state of a system of $n$ tasks can be defined as a $(2n)-tuple$ $S=C_{rem_1},...,C_{rem_n},\Omega_1, ...,\Omega_n$, where:
\begin{itemize}
\item $\Omega_i$ is the local clock of $\tau_i$, undefined before $O_i$, initialized at 0 at the time $O_i$, being reset at every period of the task. The value of $\Omega_i$ is hence in the domain $[0..T_i)$,
\item while $C_{rem_i}$ gives the remaining work to process for $\tau_i$.
\end{itemize}
The pre-state of a system of $n$ tasks can be defined as a $(2n)-tuple$ $\hat{S}=\hat{C}_{rem_1},...,\hat{C}_{rem_n},\Omega_1, ...,\Omega_n$, where:
\begin{itemize}
\item $\Omega_i$ is the same local clock as in the state $S$ of the system,
\item $\hat{C}_{rem_i}$ gives the remaining work to process for $\tau_i$ not taking the releases into account at the considered instant.
\end{itemize}
We can formalize the remaining work in state and pre-state, for any $t\geq O_i$ as:
\begin{eqnarray*}
\hat{C}_{rem_i}(O_i)&=&0 \\
C_{rem_i}(t)&=&\hat{C}_{rem_i}(t)+C_i\textnormal{ if }\Omega_i=0\\
&&\hat{C}_{rem_i}(t)\textnormal{ otherwise} \\
\hat{C}_{rem_i}(t+1)&=&C_{rem_i}(t)-1\textnormal{ if }\tau_i\textnormal{ executed on }[t..t+1) \\
&&C_{rem_i}(t)\textnormal{ otherwise}
\end{eqnarray*}
\end{definition}

\begin{definition}[Deterministic and memoryless scheduler]
\label{def:deterministic}
A scheduler is deterministic and memoryless if and only if the scheduling decision is the same for any identical state encountered.
\end{definition}

We can note that the notion of memoryless schedulers excludes schedulers taking the history into account, e.g., a scheduler that would affect a different priority to each job of the same task, depending on the job number, would be deterministic but not memoryless.

As an illustration of non-pseudo-work-conserving algorithm, we consider the task system $Sys_2$, composed of two synchronous tasks executed on a single processor: $\tau_1$ with $O_1=0$, $C_1=2$, $T_1=4$ and $D_1=5>T_1$, and $\tau_2$ with $O_2=0$, $C_2=1$, $T_2=D_2=4$. The utilization factor is $U_4=3/4$, and the hyperperiod of the system is $H=4$. Let us consider the graph of every feasible schedule of $Sys_2$ on Fig.~\ref{fig:schedgraph}. Every vertex is a reachable state, giving the values of $C_{rem_i}$. The local clocks can be obtained considering the global time given on the left side of the figure using the formula $\Omega_i=t\bmod{T_i}$. Every path is a feasible schedule of the system $Sys_2$: branching to the left corresponds to the execution of $\tau_1$, while branching to the right corresponds to the execution of $\tau_2$, and branching to the far right is inserting an idle slot. On this graph, a node marked as $cp$ is a cycle point, for example $cp_1$ and $cp'_1$ are the same state (same remaining times, same local clocks), represented twice in the figure in order to improve the readability. The cycle point $cp_1$ corresponds to a point where the load is given by the instantaneous release of all the tasks (the remaining load to process is thus $2,1$), and the global clock is $kH$ (i.e. the local clocks are both $0$), while the cycle point $cp'_1$ occuring at the instant $4+kH$ corresponds to the same state (that is why these two points are linked with a dashed line in the graph). The cycle point $cp_2$ at the date $1+kH$ corresponds to a state where the load is given by the two tasks to process ($2,1$ also), while both local clocks are given by $(1,1)$, but the system postponed the treatment by inserting an idle slot, consequently, the global clock advanced. The same state labeled $cp'_2$ can be reached at the time $5+kH$, which has the same remaining work and the same local clocks.

\begin{figure}[!t]
\centering
\includegraphics[width=3.5in]{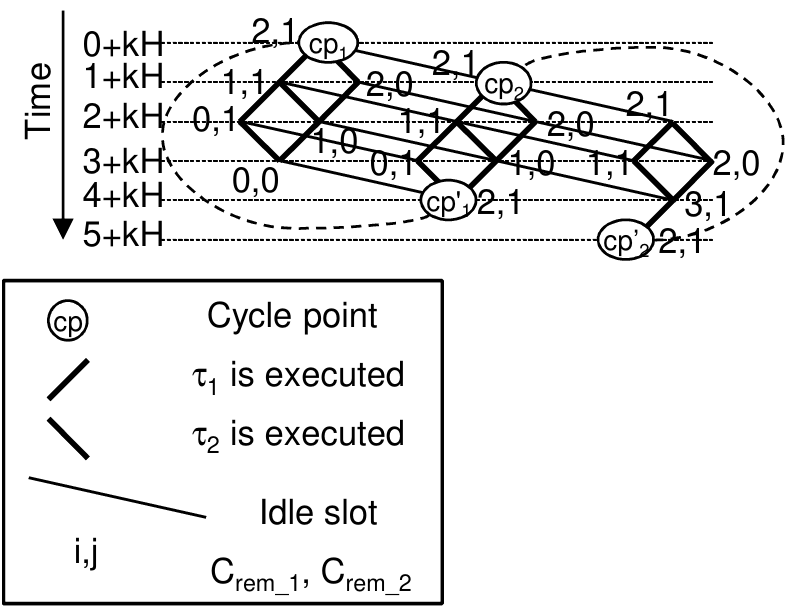}
\caption{All the feasible schedules of $Sys_2$ on one processor}
\label{fig:schedgraph}
\end{figure}

On this graph, a FPP scheduler is always giving priority to an arc compared to another one. For example, Deadline Monotonic will always choose an arc labeled $\tau_2$ over an arc labeled $\tau_1$ over an arc labeled idle slot. A fixed-job priority scheduler, like EDF, will also follow a priority scheme in the way to go down the graph. A work-conserving scheduling algorithm always chooses a task over an idle slot, while a pseudo-work-conserving algorithm allows only up to $H(1-U)$ (here one) idle slots per hyperperiod to be considered over a task. A deterministic and memoryless scheduler, when reaching twice the same node (i.e. state), makes the same scheduling decision.

Two feasible non-pseudo-work-conserving deterministic and memoryless schedules are presented in Fig.~\ref{fig:nonwcsched}, In the part $(a)$ of this figure, the cyclic part of the schedule lasts $2H$, because the execution of $\tau_1$ spills over the second hyperperiod. We can observe that, in order to obtain a cycle of two hyperperiods, the schedule is passing by $cp_1$ then $cp_2$, then $cp_1$, forcing the deterministic and memoryless schedule to behave cyclically. We can also notice that it is not possible to find a feasible deterministic and memoryless schedule of $Sys_2$ with a cyclic part of size $kH$ with $k>2$. Moreover, notice that the task $\tau_1$ is not finished in the first hyperperiod, but finished twice in the second hyperperiod of the cyclic part of the schedule.
In Fig.~\ref{fig:nonwcsched}$(b)$, we can observe a non-empty transient phase: the schedule behaves cyclically over $cp_2$, and we can notice that the state initiating the cyclic behavior of the schedule occurs at the date $2+kH$ with $k=0$ with the remaining load given by $2,0$. Note that by exploration of the paths of the graph, it was not possible to exhibit a transient phase longer than $H$ for $Sys_2$.

\begin{figure}[!t]
\centering
\includegraphics[width=3in]{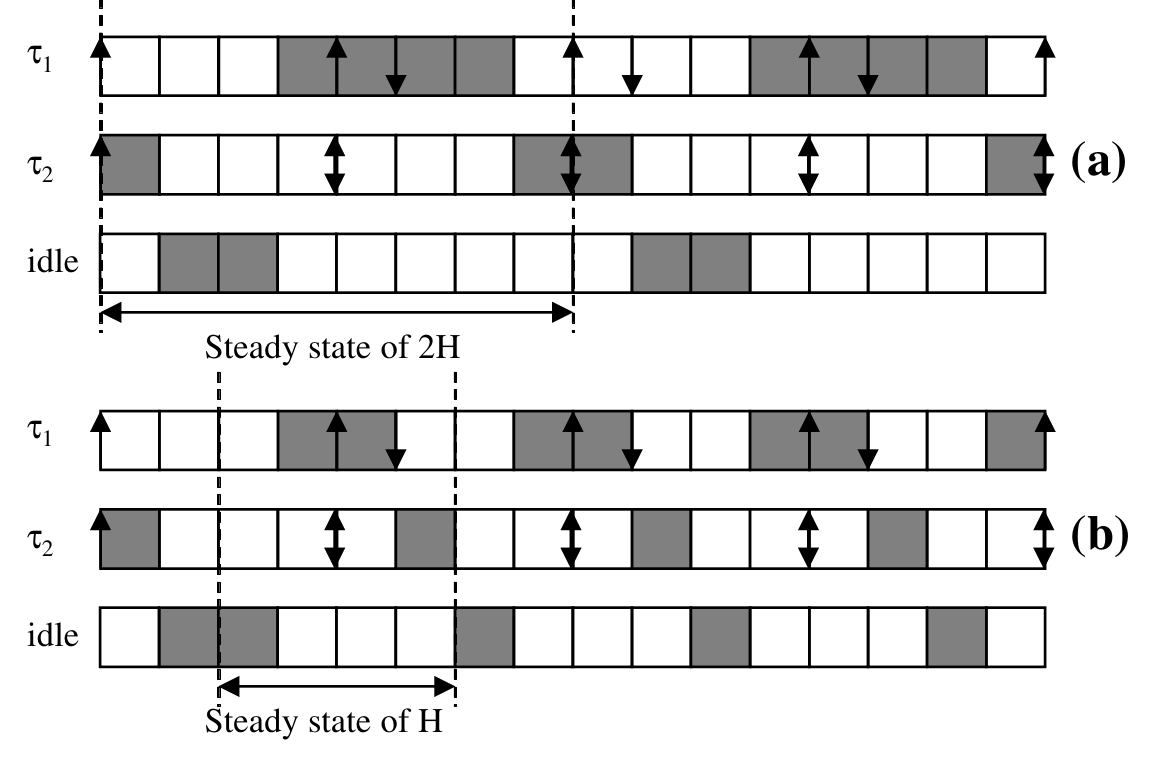}
\caption{Two feasible non-pseudo-work-conserving schedules of $Sys_2$ on one processor}
\label{fig:nonwcsched}
\end{figure}

Intuitively, bounding the cycle of schedules generated by deterministic and memoryless algorithms consists in imagining a scheduler trying to build the longest possible path in the graph of feasible schedules such that two identical are the farthest apart (because if it is reaching two identical states, then since the scheduler is deterministic and memoryless, then a cycle is reached). Note that the case of fixed-task priority schedulers and fixed-job priority schedulers like EDF, or even work-conserving schedulers, or pseudo-work-conserving schedulers, force to choose one direction over another one, forcing the schedule to behave cyclically \emph{sooner} than other algorithms, since reaching two identical states cannot be avoided. 

In a feasible graph, an infinite path has, eventually, to pass twice by the same vertex, accordingly, it is ultimately periodic. We have seen on the example given in Fig.~\ref{fig:nonwcsched}$(a)$ that the period was not necessarily the hyperperiod, but must be a \emph{multiple} of the hyperperiod, since the local clocks have to be identical.

Consider a system with two tasks $\tau_i$ and $\tau_j$ such that $D_i>T_i$ and $D_j>T_j$. The scheduler could build a schedule where $\tau_j$ is fully executed only $H/T_j-1$ in the first hyperperiod, while $\tau_i$ is fully executed $H/T_i$ times, in the second hyperperiod, $\tau_j$ is catching up (i.e. executed $H/T_j+1$ times), but $\tau_i$ is executed $H/T_i-1$ times, and in the third hyperperiod, $\tau_i$ is catching up, and $\tau_j$ is executed $H/T_j$ times, therefore the system cannot reach a cycle before three hyperperiods. We see that this leads to a length of the period of the schedule depending on the number of possible combinations of how late the tasks are allowed to be in every hyperperiod of the cyclic part of the schedule, and how they can catch up, until we reach twice a cyclic point. If a task is such that its deadline allows two of its jobs to spill over another (or several other) hyperperiods, then it is giving more possible combinations in the cyclic part of the schedule.

We can therefore imagine a very high upper bound for the size of the cyclic part of a feasible schedule, because it has to take into account any way a scheduler could possibly build a long schedule (long transient part, long cycle). We can also see on the examples in Fig.~\ref{fig:nonwcsched} that if the scheduler is pushing the transient part, then it can have an impact on the length of the cyclic part of the schedule: for $Sys_2$, any feasible schedule with a non-empty transient phase has a cyclic part lasting $H$ only.

The factor allowing a task $\tau_i$ to spill on another hyperperiod is the fact that at least one of its deadlines is not synchronized with the hyperperiod of the system. This is possible only in two cases:
\begin{enumerate}
\item \label{item:ri} the release time $O_i$ of $\tau_i$ is greater than 0, assuming that 0 is corresponding to the first release in the system,
\item \label{item:Di} the relative deadline $D_i$ of $\tau_i$ is greater than its period $T_i$.
\end{enumerate}

In order to simplify the following demonstrations, we show that the set of possible schedules for the case \ref{item:ri}) is included in the set of possible schedules for the case \ref{item:Di}).

\begin{definition}[Feasible schedule]
\label{def:feasiblesched}
Let S be a task system where tasks are defined by a first release time $O_i$, activated at a period $T_i$ and having a relative deadline $D_i$. A feasible schedule $\sigma$ is such that all the structural constraints (precedence constraints, mutual exclusions, suspensions, etc.) are respected, and the infinite set of jobs of every task $\tau_i$ is executed in its time window. We denote $s{_\sigma}(\tau_{i,j})$ (resp. $e{_\sigma}(\tau_{i,j})$) the starting date (resp. ending date) of the $j^{th}$ job of $\tau_i$ in the schedule $\sigma$. Every job is executed in its time window $[a_{i,j},d_{i,j}]$ if and only if it satisfies $s{_\sigma}(\tau_{i,j})\geq a_{i,j}$ and $e{_\sigma}(\tau_{i,j})\leq d_{i,j}$.
\end{definition}

In order to obtain a very general result on the periodicity, we define the notion of \emph{offset-independent constraints} on the tasks. This is meant to cover popular structural constraints like mutual exclusions, precedence constraints, suspension delays, non-preemptive tasks, etc.
\begin{definition}[Sub-job]
A sub-job $j$ is a part of a job in the content of a schedule.
\end{definition}

Since we consider time-based schedulers, a sub-job is thus an amount of time-units which is part of the same job. For example, we can consider every time units of a job as sub-jobs, as well as, the part between the second and the fourth time units of execution of the job.

\begin{definition}[Structural constraint]
A structural constraint is defined as a binary relation between two different sub-jobs $\cal{R}:$~$j_i \times j_k$, that may have additional arguments, like a delay. A schedule satisfies a structural constraint if the relative order of the sub-jobs and/or the required delays between the sub-jobs are met by the schedule.
\end{definition}

Some popular structural constraints are the relations \emph{excludes}, used to represent mutual exclusions between sub-jobs or non-preemptive tasks, \emph{precedes}, expressing precedence constraints between jobs, or \emph{suspends}$(t)$ where $t$ is a suspension delay, representing the suspension delay between subsequent sub-jobs of the same task, etc.

In order to introduce the next definition, it is important to note that a schedule is characterized by two parts: its content, i.e. the jobs executions on the processors, and its constraints, which are the time windows, and the structural constraints.

\begin{property}[Offset-independent structural constraint]
\label{prop:offsetindep}
Let $\cal{R}$ be a structural constraint satisfied in an original feasible schedule $\sigma$, which is characterized by an infinite set of jobs executed in their respective time window (see Def.~\ref{def:feasiblesched}). We denote ${\cal A}=\{a_1, a_2,\ldots\}$ the infinite set of activation times of the jobs. $\cal(R)$ is an \emph{offset-independent structural constraint} if this constraint is also satisfied in the schedule $\sigma'$ which is characterized by the same content as $\sigma$, but where the activation times ${\cal A}'=\{a'_1, a'_2,\ldots\}$ of the jobs can occur earlier than in the original schedule. 
\end{property}

\begin{figure}[!t]
\centering
\includegraphics[width=3.5in]{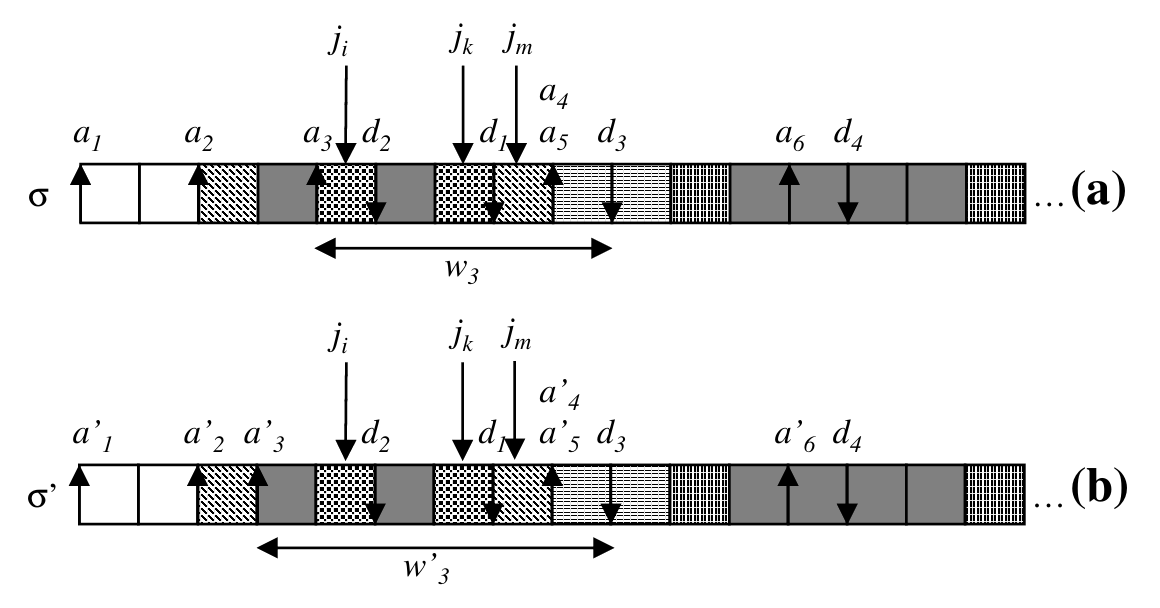}
\caption{(a) A feasible schedule $\sigma$ characterized by time windows, (b) $\sigma'$ has the same content but time windows can start earlier}
\label{fig:offsetindep}
\end{figure}

An illustration of Property~\ref{prop:offsetindep} is given in Fig.~\ref{fig:offsetindep}: in the part $(a)$, we consider a feasible schedule $\sigma$, where every job is executed in its respective time window. We focus on the window $w_3=[a_3,d_3]$, where a job, composed of the two sub-jobs $j_i$ and $j_k$, has to be executed. Now consider $\sigma'$ in part $(b)$ of the figure: the content of the schedule is the same as in $\sigma$, but a time window is larger since $a'_3 < a_3$. The temporal constraints are obviously met in $\sigma'$ if they are met in $\sigma$, because the time windows are the same or larger in $\sigma'$ than in $\sigma$. A structural constraint is offset-independent if the fact that it is respected in $\sigma$ implies that it is met also in $\sigma'$. As an example, suppose a structural constraint stating that the sub-job $j_k$ must precede the sub-job $j_m$: if this constraint is met in $\sigma$, then it is also respected in $\sigma'$ since the sub-jobs order is the same, and we can obviously see that a precedence constraint is offset-independent. Offset-independence is simply expressing the fact that, for popular structural constraint, considering a schedule independently from the release times does not have any impact on the respect of the structural constraints. We can make the same observation for other constraints such as exclusions or suspensions. Note also that some constraints are not offset-independent: for example, any constraint defined as an intermediate deadline using the activation time as time origin is not offset-independent.

\begin{lemma}
\label{lem:riDi}
Let S be a system where tasks have temporal parameters and may have offset-independent structural constraints. We denote $O_i$ the offset of the task $\tau_i$ and $D_i$ its relative deadline. Let S' be the same system, except for the release dates given by $O'_{i}=0$ and the relative deadlines $D'_{i}=D_i+O_i$. The set of feasible schedules for S is included in the set of feasible schedules for S'.
\end{lemma}
\begin{proof}
Let $\sigma$ be a feasible schedule for S, then the sequence $\sigma$ is also a feasible schedule for S' since the same offset-independent structural constraints are met (same relative order of the sub-jobs) and $s{_\sigma}(\tau_{i,j})\geq 0+jT_i$ and $e{_\sigma}(\tau_{i,j})\leq 0+jT_i+D^{'}_i$.
\end{proof}

The underlying idea behind Lemma \ref{lem:riDi} is that the time window allocated to every job in S' is including the time window allocated to every job in S. Since we are interested to any deterministic and memoryless scheduling algorithm, we see that focusing only on synchronous task systems with an arbitrary deadline cannot reduce the possibilities for a scheduling algorithm to delay its cyclic part. For this reason, an upper bound on the case where deadlines are arbitrary is also an upper bound for asynchronous task systems.
\begin{lemma}
\label{lem:statesync} For synchronous task systems, if two pre-states are identical, then the scheduling decision of a deterministic and memoryless scheduler is the same.
\end{lemma}
\begin{proof}
Following the definition of deterministic and memoryless schedulers, if two states are identical, then the scheduling decision is the same. This Lemma states that it is sufficient to consider the pre-state in the case of synchronous systems. Indeed, if two pre-states are identical, their local clocks are the same (and so do the clocks of the corresponding states). Considering $t$ and $t'$ the respective time instants where $\hat{S}$ and $\hat{S'}$ occur, it is easy to see that $t'=t+kH, k\in\mathbb{N}$, which are the only possible solutions such that every local clock, all starting at the instant $0$ (the system is synchronous), are the same. If the values of the remaining work are the same in two pre-states $\hat{S}$ and $\hat{S'}$, then the values are also the same for the corresponding states $S$ and $S'$ because giving Def.~\ref{def:state}, we have $C_{rem_i}(t)=\hat{C}_{rem_i}(t)+C_i$ if $\Omega_i=0$, and $C'_{rem_i}(t)=\hat{C'}_{rem_i}(t')+C_i$ if $\Omega'_i=0$. Since $\Omega_i=\Omega'_i$, and $\hat{C}_{rem_i}(t)=\hat{C'}_{rem_i}(t')$, then $C_{rem_i}(t)=C'_{rem_i}(t')$, and $S=S'$.
\end{proof}

It follows from Lemma \ref{lem:statesync} that we can only focus on the pre-states to prove the periodicity of synchronous task systems.

\begin{lemma}
\label{lem:general}
Any feasible schedule of a synchronous task system generated by a deterministic and memoryless scheduler reaches a cycle at or prior to $\left(\prod_{i=1}^{n}{((D_i-T_i)_0+1)}\right)H$, where $(a)_0=\max(a,0)$.
\end{lemma}
\begin{proof}
Note that the pre-state at the time 0 is given by $\hat{S}_0^0$ in Fig.~\ref{fig:stateproof1}. On this figure, since only hyperperiods are considered, the pre-states can be represented only by the values of $C_{rem_i}$, every local clock being null. In order to prove the Lemma, we will show that the number of distinct pre-states for every hyperperiod $kH, k\in\mathbb{N}^+$, in any feasible schedule, is bounded by $\prod_{i=1}^{n}{((D_i-T_i)_0+1)}$.
\begin{itemize}
\item Constrained deadlines case: if every task has a constrained deadline (i.e., $D_i \leq T_i$), then if the pre-state reached at the date $H$ is not $0$ for every remaining time, then the schedule is not feasible. Indeed, every job started during the first hyperperiod has to be finished before the end of this hyperperiod. As a result there is only one possible pre-state at the date $H$ for any feasible schedule, which is identical to the initial state. Hence, any feasible schedule is cyclic over the interval $[0,H)$, showing the Lemma for this case.
\item Case of only one task, $\tau_i$, having a deadline greater than their period. We first give the proof sketch for $D_i=T_i+1$. In any feasible schedule there is at most one time unit of the $(H/T_i)^{th}$ job of $\tau_i$ that remains at the time instant $H$, else the system cannot be feasible, while every other job released in the first hyperperiod has to be finished since $D_j\leq T_j \forall j\neq i$. Therefore the only possible pre-states of the system in a feasible schedule at the date $H$ can be defined by $\hat{S}_0^H$ and $\hat{S}_1^H$ in Fig.~\ref{fig:stateproof1}. Note that $\hat{S}_0^H$ is the same as $\hat{S}^0_0$, and so if the schedule reaches this state, then it is behaving cyclically from this point: the schedule $[0,H)$ will be repeated infinitely. If the system is in $\hat{S}_1^H$, then consider the schedule at the time $2H$: there again, only  two possible pre-states can be part of a feasible schedule, $\hat{S}_0^{2H}=\hat{S}_0^0$ and $\hat{S}_1^{2H}=\hat{S}_1^H$. If the system is in the pre-state $\hat{S}_0^{2H}$ then the schedule behaves cyclically over the interval $[0,2H)$; else the schedule has a transient phase on $[0,H)$ (from pre-state $\hat{S}_0^0$ to $\hat{S}_1^H$) followed by a cyclic phase on $[H,2H)$ (from $\hat{S}_1^H$ to $\hat{S}_1^{2H}$). The maximal simulation duration is hence $2H$, proving the theorem for one task having $D_i=T_i+1$.

Now suppose that $D_i=T_i+k$ with $k$ an arbitrary finite positive integer. If we name $\hat{S}_j^{pH}$ any reachable pre-state in a feasible schedule where $0\leq j\leq k$ gives the remaining work to process for $\tau_i$ at the date $pH$, with $p$ a positive integer, it is obvious that there are only $k+1$ possible different pre-states. As a consequence, the possible cyclic behaviors of any feasible schedule are bounded by $(k+1)H$. Any combination of a transient phase lasting over $[0,qH)$ followed by a cyclic phase over $[qH,rH)$ with $0\leq q < k$, $r\geq q+1$ and $r\leq k$ can be a feasible memoryless and deterministic schedule.
\begin{figure}[!t]
\centering
\includegraphics[width=3in]{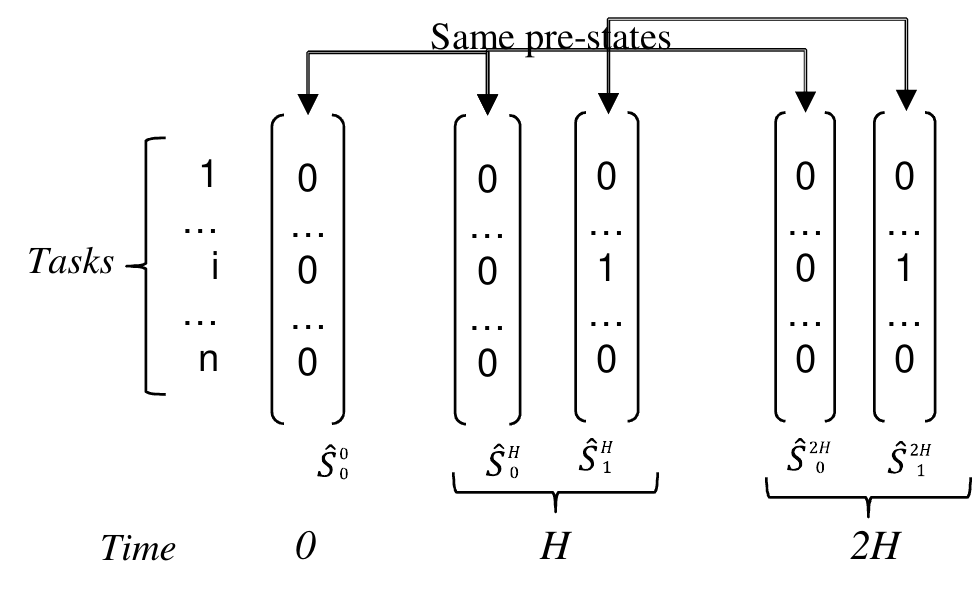}
\caption{States that can be reached in a feasible schedule at times $0$, $H$ and $2H$ for $D_i-T_i=1$}
\label{fig:stateproof1}
\end{figure}
\item If several tasks have a deadline greater than their period, then we could represent the pre-states that may be reached by a feasible schedule each hyperperiod $H$ as a $n$-dimensional matrix given by the Cartesian product of pre-states where each task can be delayed by an amount between $0$ and $(D_i-T_i)_0$. The number of elements of this matrix is therefore $\prod_{i=1}^{n}{((D_i-T_i)_0+1)}$. As a result, it is impossible for a feasible schedule not to have reached two identical pre-states after $\left(\prod_{i=1}^{n}{((D_i-T_i)_0+1)}\right)H$ time units.
\end{itemize}
\end{proof}

\begin{theorem}
\label{th:general}
Any feasible schedule of a task system generated by a deterministic and memoryless scheduler reaches a cycle at or prior to $\left(\prod_{i=1}^{n}{((O_i+D_i-T_i)_0+1)}\right)H$.
\end{theorem}
\begin{proof}
Since the result stands for synchronous systems with arbitrary deadlines (Lemma \ref{lem:general}), following Lemma \ref{lem:riDi}, we know that any feasible schedule for an asynchronous system $S$ is a feasible schedule for a synchronous system $S'$ such that $O'_i=0$ and $D'_i=O_i+D_i$, hence any feasible schedule for $S$ reaches a cycle at or prior to $\left(\prod_{i=1}^{n}{((D'_i-T_i)_0+1)}\right)H$. Substituting $D'_i$ by $O_i+D_i$ we obtain the Theorem.
\end{proof}

\section{Discussion}
\label{sec:discuss}
In this section, we use the notation $[a,b)$ to characterize a schedule over the time interval starting from $a$ included and ending at $b$ excluded. We consider the star operator as the repetition of its preceding interval, and $[a,b)[c,d)*$ represents the schedule over the interval $[a,b)$ followed by the schedule over $[c,d)$ repeated in cycle.

\paragraph*{Application of the main result}
If we use Theorem~\ref{th:general} on $Sys_1$, we obtain an upper bound of $(3+1)\times 1\times 1\times H=4H$ for the simulation interval. We see that, for LRPTF on Fig.~\ref{fig:lrptf} we have an infinite feasible schedule $[0,H)^*$, while global-EDF on Fig.~\ref{fig:edfmulti} gives a feasible schedule $[0,2H)[2H,3H)^*$. The states reached by global-EDF at each hyperperiod are $(0,0,0,0,0,0)$ at the origin, $(0,0,1,0,0,0)$ at the time $H$, $(0,0,2,0,0,0)$ at $2H$, $(0,0,2,0,0,0)$ at $3H$. We can, as an example, build a feasible schedule lasting $[0,4H)^*$ starting at the state $(0,0,0,0,0,0)$, and then passing by the states $(0,0,1,0,0,0)$ at $H$, $(0,0,2,0,0,0)$ at $2H$, $(0,0,3,0,0,0)$ at $3H$, and $(0,0,3,0,0,0)$ at $4H$, as illustrated on Fig.~\ref{fig:4H}. The scheduling algorithm used to generate such a schedule is not corresponding to any popular scheduling algorithm, but we can imagine a deterministic and memoryless scheduling algorithm, giving this schedule, defined by an array indexed by a state of a system giving for any possible state a scheduling decision.

\paragraph*{Comparison with other existing bounds}
In order to compare our bound to the bound provided for the case of FPP schedulers in~\cite{CucuG07}, we consider a simple system of two tasks $\tau_1$ and $\tau_2$ with the same period $T_1=T_2=8$, and offsets and deadlines given by $O_1=1, D_1=7, O_2=0, D_2=8$, and we consider a FPP scheduler assigning a higher priority to $\tau_1$ than to $\tau2$. Theorem~\ref{th:general} gives a simulation interval $[0,8)$, while the bound given in~\cite{CucuG07} (see Eq.~\ref{eq:sn2}) gives $[0,24)$. In this case, since the deadlines are lower than the periods, we could also use the upper bound given in~\cite{Cucu06} (see Eq.~\ref{eq:sn}), and obtain the simulation interval $[0,16)$.

If we consider a different system with $O_1=1, D_1=7, T_1=12, O_2=0, D_2=9, T_2=8$, then the simulation intervals are $[0, 96)$ for Theorem~\ref{th:general}, and still $[0,24)$ for~\cite{CucuG07}, and cannot be calculated with~\cite{Cucu06}, because $D_1>T_1$. We can see that the bounds are not comparable, therefore, in the case where several upper bounds could be applied, the minimal value giving a simulation interval upper bound should be chosen.

Note that if the tasks were involving any structural constraint as mutual exclusions, precedence constraints, suspension delays, or non preemptive tasks, Theorem~\ref{th:general} would still hold, while the other periodicity results concerning multiprocessor systems are not applicable.

\begin{figure}[!t]
\centering
\includegraphics[width=3.5in]{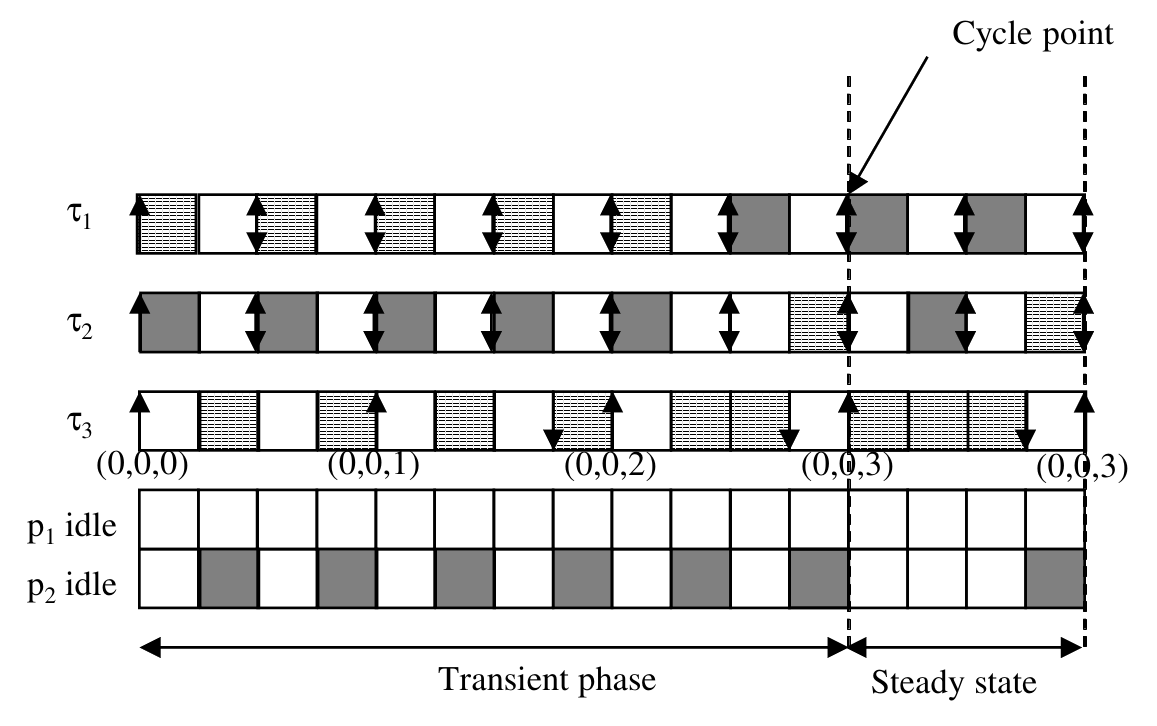}
\caption{Schedule lasting $4H$ generated by a deterministic and memoryless scheduler for $Sys_1$ on two processors}
\label{fig:4H}
\end{figure}

\paragraph*{Tightness}
Our bound is safe, but not tight, as illustrated in the following example. Let a system be composed of two synchronous independent tasks $\tau_1$ and $\tau_2$, executed on a single processor, such that $D_1=T_1+1$ and $D_2=T_2+1$. Theorem~ \ref{th:general} is giving an upper bound of $4H$ for the cycle, because the pre-states that can be reached at each hyperperiod are given by $(0,0,0,0)$, $(0,1,0,0)$, $(1,0,0,0)$ and $(1,1,0,0)$. But clearly, if both tasks have a remaining processing time of one time unit, with zero laxity (both deadlines happen one time unit after the considered hyperperiod), then the schedule cannot be feasible. As a consequence, in this case, the longest feasible schedule without reaching twice the same state is constrained to the time interval $[0,3H)$. In general, a test like a demand bound function could be used to check if the states obtained by the Cartesian product of the possible lateness of the tasks can lead to a feasible schedule or not in order to reduce the bound.

\section{Conclusion}
The problem tackled in this paper is the periodicity problem for feasible schedules produced by any deterministic and memoryless scheduler, in uniprocessor and multiprocessor cases, for any offset-independent structural constraints (mutual exclusions, precedence constraints, self-suspension, non-preemptive tasks, etc.). The result concerning the periodicity of schedules is the most general result ever proposed in the context of uniprocessor scheduling as well as in the context of identical multiprocessor systems, since it concerns any deterministic and memoryless scheduler, arbitrary deadlines, and dependent task systems.
We have proven that for this problem, the synchronous arbitrary deadline case is a generalization of the asynchronous case. This intermediate result has a major impact on the relative simplicity of the proof of the main theorem. Lemma \ref{lem:riDi} could also be used by itself to improve existing or future periodicity results concerning specific scheduling algorithms. Then we have shown that the cycle is reached for any feasible schedule at most at the time $\left(\prod_{i=1}^{n}{((O_i+D_i-T_i)_0+1)}\right)H$. This result might be improved if we take into account the local feasibility of the tasks, but we believe that the applicability of the upper bound would be weakened by the difficulty to handle it in this extended form. 

We also want to stress the fact that our result is an upper bound for any deterministic and memoryless scheduler, therefore it may be improved for specific scheduling algorithms. As an example, specific bounds concerning FPP schedulers like in~\cite{CucuG07}, can offer a lower or larger upper bound than ours depending on the case. The best known bound would then to be considered, for such a specific case (FPP, independent tasks), as the minimal value of the two upper bounds.

In the future, we plan to extend this result to uniform and unrelated multiprocessor platforms. We also plan to improve existing bounds for specific scheduling algorithms using our intermediate result, the Lemma \ref{lem:riDi}.

\section*{Acknowledgments}
We wish to thank Pascal Richard for our very inspiring discussions.
\bibliographystyle{acm}
\bibliography{biblio}



\end{document}